\documentclass{llncs}
\usepackage{times}

\usepackage{amssymb,amsfonts,amsmath}
\usepackage {color}
\usepackage{lineno}
\usepackage{verbatim}
\usepackage{fancybox}
\usepackage{soul}
\usepackage{graphicx}

\newcommand{\oracle}{{\mathcal{O}}}

\newcommand{\redClusters}{S_\ell}
\newcommand{\blueClusters}{S_r}
\newcommand{\HVD}[1]{\ensuremath{\textup{\textsf{HVD}}(#1)}}
\newcommand{\FCVD}[1]{\ensuremath{\textup{\textsf{FCVD}}(#1)}}
\newcommand{\VD}[1]{\ensuremath{\textup{\textsf{VD}}(#1)}}
\newcommand{\redHVD}{\ensuremath{\textup{\textsf{HVD}}(\redClusters)}}
\newcommand{\blueHVD}{\ensuremath{\textup{\textsf{HVD}}(\blueClusters)}}
\newcommand{\mergeCurve}{\sigma}

\newcommand{\hreg}[1]{\ensuremath{\textup{\textsf{hreg}}}(#1)}
\newcommand{\fcreg}[1]{\ensuremath{\textup{\textsf{fcreg}}(#1)}}
\newcommand{\df}[1]{\ensuremath{\textup{\textsf{d}}_\textup{\textsf{f}}}(#1)}
\newcommand{\hregs}[2]{\ensuremath{\textup{\textsf{hreg}}_{#1}(#2)}}
\newcommand{\freg}[2]{\ensuremath{\textup{\textsf{freg}}_{#1}(#2)}}
\newcommand{\FVD}[1]{\ensuremath{\textup{\textsf{FVD}}(#1)}}
\newcommand{\fskel}[1]{\ensuremath{{\mathcal{T}}(#1)}}
\newcommand{\bh}[1]{\ensuremath{\textup{\textsf{b}}_\textup{\textsf{h}}}(#1)} 
\newcommand{\theProblem}{{RB-Preprocessing problem}}
\newcommand{\hrego}{\ensuremath{\textup{\textsf{hreg}}}()}
\newcommand{\frego}{\ensuremath{\textup{\textsf{freg}}}()}

\title{Searching edges in the overlap of two plane graphs\thanks{
Research partially completed while J. I. was on sabbatical at the Algorithms Research Group of the D\'epartement d’Informatique 
at ULB with support from a Fulbright Research Fellowship, F.R.S.-FNRS, and NSF grants CNS-1229185, CCF-1319648, and CCF-1533564. E. K. was supported in part by F.R.S.-FNRS, and by SNF project P2TIP2-168563 under SNF Early PostDoc Mobility program; S. L. is directeur de recherches du F.R.S.-FNRS.
}}

\author{John Iacono\inst{1}
\and Elena Khramtcova\inst{2}
\and Stefan Langerman\inst{2}}

\institute{Department of Computer Science and Engineering, New York University 
\\  New York, USA, \email{iacono@nyu.edu} 
\and
Computer Science Department, Universit\'e libre de Bruxelles (ULB) 
\\  Brussels, Belgium, 
\email{{\{elena.khramtcova,stefan.langerman\}@ulb.ac.be}}} 

\authorrunning{J. Iacono, E. Khramtcova, and S. Langerman}

\begin{document}
\maketitle

\begin{abstract}
Consider a pair of plane straight-line 
graphs 
whose edges are colored red and blue, respectively, and let $n$ be the total  
complexity of both  graphs.   
We present a $O(n\log{n})$-time $O(n)$-space technique to preprocess such a pair of graphs, that enables 
efficient searches among the red-blue intersections along edges of one of the graphs. 
Our technique has a number of applications to geometric problems. 
This includes: 
(1) a solution to the \emph{batched red-blue search} problem~[Dehne et al. 2006] in $O(n\log{n})$ queries to the oracle; 
(2) an algorithm to compute  the maximum vertical distance  between a pair of 3D polyhedral terrains,   
one of which is convex, in $O(n\log{n})$ time, where $n$ is the total complexity of both terrains;      
(3) an algorithm to construct the Hausdorff Voronoi diagram 
	of a family  of point clusters in the plane in  
$O((n+m)\log^3{n})$ time and $O(n+m)$ space,   
	where $n$ is the total number of points in all clusters and $m$  is the number of \emph{crossings} 
	between all clusters; 
(4) an algorithm to construct the farthest-color Voronoi diagram of the corners of $n$ disjoint axis-aligned rectangles in $O(n\log^2{n})$ time; 
(5) an algorithm to solve the stabbing circle problem for $n$ parallel line
 segments in the plane
 in optimal $O(n\log{n})$ time.
All these results are new or improve on the best known algorithms.  
\end{abstract}

\section{Introduction}\label{sec:intro}

Many geometric algorithms have 
subroutines that involve investigating intersections between two plane graphs, often assumed being colored 
red and blue respectively. 
Such subroutines differ in the questions that are asked about the red-blue intersections. 
The most well-studied questions  are 
to report all red-blue intersections or to count them. 
It is shown how to report all the intersections  
in optimal $O(n\log{n} +k)$ time and $O(n)$ space~\cite{chan94,CEGS94,ms1988,ms01,ps94},
 where $n$ is the total complexity of 
both graphs, and $k$ is the size of the output. 
Note that $k$ may be $\Omega(n^2)$. 
Counting the red-blue intersections can be carried out in $O(n\log{n})$ time and  $O(n)$ space~\cite{CEGS94,ms01}. 

In this paper, we consider the situation  where  one wants to \emph{search} the red-blue intersections, though avoiding 
to compute all of them. Problems of this type 
  appear as building blocks in diverse geometric algorithms. The latter include: distance measurement between 
polyhedral terrains~\cite{CEGS94}, motion planning
\cite{Guibas1989}, construction of various generalized Voronoi diagrams 
(divide-and-conquer~\cite{dmt06,CEGGHLLN11} or randomized incremental~\cite{CKLP16} construction). 
Therefore solving such problems efficiently is of high importance.

Often it is guaranteed that each red edge contains at most one sought red-blue intersection, and an oracle is provided,
that, given a red-blue intersection,  is able to quickly determine 
to which side of that intersection  the sought intersection lies 
along the same red edge (see Section~\ref{sec:batched}
for more details on this setting).
A particular case, when the red graph consists of a unique edge, 
appeared under the name of \emph{segment query} in the randomized  incremental construction algorithm for the Hausdorff Voronoi diagram~\cite{CKLP16}, 
and under the name of \emph{find-change} query in an algorithm to solve the stabbing circle problem for a set of line segments~\cite{CKPSS17}. 
If the blue graph is a tree, it can be preprocessed  in $O(n\log{n})$ time
using the \emph{centroid decomposition}~\cite{aronov2006,CKLP16}. Centroid decomposition supports 
segment (or find-change) queries for arbitrary line segments, requiring only $O(\log n)$ queries to the oracle~\cite{CKLP16,ckpss16-eurocg}.
If the blue graph is not a tree, then  in $O(n\log{n})$ time it can be preprocessed for point location, and a nested point location 
along the red edge is performed, which requires $O(\log^2 n)$ queries to the oracle~\cite{CEGGHLLN11,CKPSS17}.
For two general plane straight-line graphs (where the red graph is not necessarily one edge) 
the problem is called \emph{batched red-blue intersection problem} (see Problem~\ref{prob:batched}). 
It was formulated in Dehne et al.~\cite{dmt06}, and solved in
 $O(n\log^3{n})$  time and $O(n\log^2{n})$ space~\cite{dmt06} using \emph{hereditary segment trees}~\cite{CEGS94}. 
 However, this is optimal in neither time nor space. 

We present a data structure that provides a clear interface for efficient searches
for red-blue intersections along a red edge.   
Our data structure can be used to improve the above
 result~\cite{dmt06} (see Section~\ref{sec:batched}), which includes an improvement on segment (or find-change) queries in plane straight-line graphs. 
Our data structure can also handle more general search problems, e.g., a 
setting when a red edge may have more than one sought red-blue intersection on it. 
Below we state our result and its applications.

\subsection{Our result}
Let $R$, $B$ be a pair of plane straight-line\footnote{Our technique can be trivially generalized to apply to $x$-monotone 
pseudoline arcs in place of straight-line edges of the graphs.} graphs. 
We address the following problem. 

\begin{problem}[\theProblem]
\label{prob:rb}
Given graphs $R,B$,  
construct a data structure that for each edge $e$ of $R$ stores implicitly the  intersections between $e$ and the edges of $B$
sorted according to the order, in which these intersections appear along $e$.  
Let $T_e$ be a perfectly balanced binary search tree
 built on the sorted sequence of intersections along $e$. 
The data structure should 
answer efficiently  the following \emph{navigation queries} in $T_e$:
\begin{itemize}
\item Return  
the root of $T_e$;
\item Given a non-root node of $T_e$, return the parent of this node; 
\item Given a non-leaf node of $T_e$, return the left (or
the right) child of this node.
\end{itemize}
\end{problem}

We provide a solution to the \theProblem, where each of the navigation queries can be answered in $O(1)$ time, 
and
constructing the data structure requires $O(n\log n)$ time and $O(n)$ space,
where $n$ is the total number of vertices and edges in both $R$ and $B$ (see Section~\ref{sec:technique}).

The resulting data structure allows for fast searches 
for \emph{interesting} intersections between edges of $R$ and the ones of $B$. We note that the notion of {interesting} 
is external to the data structure: It is not known  at the time of preprocessing, 
but rather guides the searches on the data structure after it is built.
In particular, for the input graphs $R$ and $B$, the data structure is always the same, while interesting intersections can be defined in several ways, 
which of course implies that the searches may have different outputs. 

Our preprocessing  technique 
can be applied to a number of geometric problems. 
We provide a list of applications, which is not exhaustive.
For each application, we show how to reduce the initial problem 
to searching for interesting red-blue intersections, 
and how to navigate the searches, that is, how to decide, which subtree(s) of the current node of the (implicit) tree to search. 
Using our technique we are able to make the  
contributions listed below, 
and we expect it to be applicable to many more problems.   

\begin{enumerate}
\item The \emph{batched red-blue search problem}~\cite{dmt06} for a pair of segment sets
  can be solved in $O(n)$ space and $O(n\log{n})$ queries to the oracle, where $n$ is the total number of segments in both sets (see Section~\ref{sec:batched}). 
The problem is 
as follows. 
Given are: (1) two sets of line segments in the plane (colored red and blue, respectively), 
where no two segments in the same set intersect;   
and (2) an oracle that, given a point $p$ of intersection between a red segment $r$ and a blue segment $b$, 
determines to which side of segment $r$ with respect to point $p$ the {interesting red-blue intersection} lies.
It is assumed that each segment contains at most one interesting intersection. 
The \emph{batched red-blue search problem} is to find all interesting red-blue intersections. Our solution is an improvement on the one of Dehne~et al.~\cite{dmt06}
 which requires $O(n\log^2{n})$ space and $O(n\log^3{n})$ queries to the oracle. 

\item The maximum vertical distance between a pair of 
	polyhedral terrains, one of which is convex, can be computed in $O(n\log{n})$ time and $O(n)$ space (see Section~\ref{sec:distance}).
Previously, a related notion of the {minimum} vertical distance between a pair of non-intersecting polyhedral terrains was considered, 
and it was shown how to find it in $O(n^{4/3 + \epsilon})$ time and space for a pair of general polyhedral terrains~\cite{CEGS94}, 
in $O(n\log{n})$ time for 
one convex and one  general terrain~\cite{Z97},  and in $O(n)$ time for two convex terrains~\cite{Z97}. 
Our technique yields an alternative solution for the second case within the same time bound as in~\cite{Z97}.
The  maximum distance for non-intersecting polyhedra can be found by the above methods~\cite{CEGS94,Z97}, 
however it is different from the minimum distance  
for intersecting polyhedra: 
		asking about the former is still interesting, while the latter is trivially zero. 
\item The Hausdorff Voronoi diagram of  
a family of point clusters in the plane can be constructed in $O((n+m)\log^3{n})$ time,
	where $m$ is the total number of pairwise \emph{crossings} of the clusters
(see Sections~\ref{sec:hvd} and~\ref{sec:hvd-arb}).
 		Parameter $m$ can be $\Theta(n^2)$, but is small in practice~\cite{pl04,p04}. 
There is a deterministic algorithm  to compute the diagram in  $O(n^2)$ time~\cite{EGS89}. 
All other known deterministic algorithms~\cite{p04,pl04} have a running time that depends on parameters of the
		input, that cannot be bounded by a function of $m$.\footnote{ The algorithms have time complexity 
		respectively  $O(M + n \log^2 n + (m + K)\log n)$ and 
		$O( M' + (n+m+K')\log n)$, 
		where parameters $M,M',K,K'$ reflect the number of pairs of clusters such
		that one is enclosed in a certain type of enclosing circle of the other.} 
		Each of them may take $\Omega(n^2)$ time 
		even if $m=0$.  
There is a recent randomized algorithm with 
expected time complexity $O((m + n \log n) \log n))$~\cite{kp16a}. 
For a simpler case of non-crossing clusters ($m=0$), 
the diagram can be computed in deterministic $O(n\log^5{n})$ time\footnote{The time complexity claimed in~\cite{dmt06} is $O(n\log^4{n})$. See the discussion in
 Section~\ref{sec:hvd}.}~\cite{dmt06},  
or in expected  $O(n \log^2{n})$ time~\cite{CKLP16,kp16a}. 
Thus our algorithm is the best deterministic algorithm for 
the case of small number of crossings. 
The time  complexity of our algorithm is subquadratic in 
$n$ and $m$ and depends only on them,  
unlike any  previous deterministic algorithm. 
 
\item The farthest-color Voronoi diagram for a family of $n$ point clusters, where each cluster is all the  
corners of an  axis-aligned 
rectangle,\footnote{A cluster is either the four corners of a non-degenerate axis-aligned rectangle, 
or the two endpoints of a horizontal/vertical segment, or a single point.} and these rectangles are pairwise disjoint,
 can be computed in $O(n\log^2{n})$ time and $O(n)$ space (see Section~\ref{sec:fcvd}). 
Previous results on the topic are as follows. 
For  arbitrary point clusters, the diagram may have 
complexity $\Theta(n^2)$  and 
can be computed in $O(n^2)$ time and space~\cite{ahiklmps01,EGS89}, where $n$ is the total number of points in all clusters. 
When clusters are  pairs of endpoints of $n$ parallel line segments, 
the diagram  has $O(n)$ complexity 
and can be constructed in $O(n\log{n})$ time and $O(n)$ space~\cite{CKPSS17}. 
In this paper, we broaden the class of inputs, for which the diagram can be constructed in subquadratic time.
We also show that the complexity of the diagram for such inputs is $O(n)$.  

\item The stabbing circle problem for line segments in the plane can be solved in time $O(\mathcal{T}_{\HVD{S}} +\mathcal{T}_{\FCVD{S}}+(|\HVD{S}|+|\FCVD{S}|+m)\log n)$, where $|\HVD{S}|$ and $|\FCVD{S}|$ denote respectively the complexity of  
	the Hausdorff and the farthest-color Voronoi diagram of the pairs of endpoints of segments in $S$, 
	 $\mathcal{T}_{\HVD{S}}$ and $\mathcal{T}_{\FCVD{S}}$ denote the time to compute these diagrams, and $m$ 
is a parameter reflecting the number of ``bad'' pairs of segments in $S$\footnote{See Section~\ref{sec:stabbing} for the definition of $m$.} (See Section~\ref{sec:stabbing}). 
	If all segments in $S$ are parallel to each other, the stabbing circle problem can be solved in optimal 
		$O(n\log{n})$ time and $O(n)$ space. 
This is an improvement over the recent   $O(\mathcal{T}_{\HVD{S}} +\mathcal{T}_{\FCVD{S}}+(|\HVD{S}|+|\FCVD{S}|+m)\log^2 n)$
		time technique for general segments, which yielded an  $O(n\log^2{n})$ time algorithm for parallel segments~\cite{CKPSS17}. 
\end{enumerate}

\section{The technique to preprocess a pair of graphs}
\label{sec:technique}
Suppose we are given 
two plane straight-line graphs $R$ and $B$, and let $n$ be the total number of vertices and edges in both $R,B$. 
We assume that no two vertices of the graphs have the same $x$ coordinate. 
In this section, it is more convenient to  treat $R$ and $B$ as two sets of line segments in the plane, where 
the segments in $R$ are colored red, and the ones in $B$ are colored blue. 
No two segments of the same color intersect, although they may 
share an endpoint.

Our preprocessing technique consists of three phases. 
In the first phase, we invoke an algorithm that finds the intersections between the edges of the two graphs (see Section~\ref{sec:mantler-snoyeink}). 
After that, in the second phase, we build a linearized \emph{life table} 
for the red segments (see Section~\ref{sec:ds}).  
Finally  we sweep the life table with a line, which provides us the 
resulting data structure (see Section~\ref{sec:resulting-ds}).

\subsection{Finding red-blue intersections}
\label{sec:mantler-snoyeink}
For the sets $R$ (red) and $B$ (blue), we need to find all the intersections between segments of different color, i.e., all the red-blue intersections. 

It is known how to  count the red-blue intersections
in optimal $O(n\log{n})$ time~\cite{CEGS94,ms01,ps94}, or 
report them in optimal $O(n\log{n} +k)$ time~\cite{chan94,CEGS94,ms1988,ms01,ps94}, 
where $k$ is the total number of the intersections.  
The space requirement of each of these algorithms is $O(n)$.
The algorithm by Mantler and 
Snoeyink~\cite{ms01} 
processes the red-blue intersections in batches (called \emph{bundle-bundle} intersections). 
In $O(n\log{n})$ time and $O(n)$ space it can implicitly discover all the red-blue
 intersections, without reporting every one of them individually. 
The latter feature  is useful for our technique, 
therefore we invoke the Mantler-Snoeyink algorithm in its first phase.
We summarize the algorithm below. 

To describe the algorithm, we need to define the following key notions: 
the \emph{witness} of a (bichromatic) segment  intersection, 
a \emph{pseudoline at time $i$}, and 
a (monochromatic) \emph{bundle} of segments at time $i$.

\begin{figure}
\centering
\includegraphics{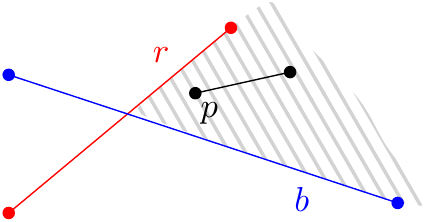}
\caption{Two line segments, $r$ and $b$, the closed right wedge formed by them, and the witness $p$ of their intersection}
\label{fig:witness}
\end{figure}

Given a red segment $r$ that intersects a blue segment $b$, the \emph{witness} of their intersection  
is the leftmost of the endpoints 
of segments in $S$ that are contained in the closed right wedge formed by $r$ and $b$. 
The closed right wedge formed by $r$ and $b$ is the intersection of two closed right halfplanes:
the one bounded by the line through $r$, and  the one  bounded by the line through $b$,  
see the shaded area in Figure~\ref{fig:witness}.                                      
Note that 
the witness always exists: it 
may be an endpoint of a segment different from $r$ or $b$ (as in Figure~\ref{fig:witness}), or it may be an endpoint of either $r$ or $b$. 

Let $n'$ be the total number of distinct endpoints of the segments in $R$ and $B$. 
Let $p_1, p_2, \ldots, p_{n'}$ denote the sequence of these endpoints
 in the order of increasing $x$ coordinate.  The  basis of the Mantler-Snoeyink algorithm can be formulated as follows.

\begin{lemma}
\label{lemma:pseudoline}
For each $i, 1 \leq i \leq n'$ there is a $y$-monotone curve $\ell_i$ that passes through point $p_i$, 
and subdivides the plane into two open regions (the left and the right one), such that 
all the points  $p_j, j<i$, and all the   red-blue intersections witnessed by the points $p_j, j \leq i$ are contained in the left region, 
and all the points $p_k, k > i$ together with the intersections witnessed by them are contained in the right region, and $\ell_i$ intersects each segment in $R$ or in $B$ at most once.  
\end{lemma}

We call such curve $\ell_i$ a \emph{pseudoline at time $i$}. 
Figure~\ref{fig:exe} shows a pseudoline at time~7, i.e.,  $\ell_7$, in dashed black lines.
Note that $\ell_7$ cannot be replaced by a vertical straight line, because it must pass through the point 7, and to the left of the intersection point 
 between the segments $r_4$ and $b_4$, and the latter point lies to the left of the former one. 

A \emph{blue bundle} 
at time $i$ is a maximal contiguous sequence of  blue segments 
that intersect the pseudoline $\ell_i$.\footnote{This definition  
can be seen as a generalization of the one 
of \emph{single-edge bundles} in Mount~\cite{M87}.} See Figure~\ref{fig:exe}, right. 
A \emph{red bundle} is defined analogously. 

The algorithm can be seen as a topological sweep with a pseudoline, where  
the only events are the endpoints $p_1, \ldots, p_{n'}$ of the segments in $R$ and $B$.
The sweepline at each moment $i$ is a pseudoline $\ell_i$
 such as defined in Lemma~\ref{lemma:pseudoline}. 
The sweepline status structure maintains all the red and blue bundles that intersect the 
 current sweepline. 
The sweepline status consists of (1) a balanced binary tree for each bundle, 
supporting insertion, deletion of segments, 
and a query for the topmost and the bottommost segment in the bundle,
 (2) a doubly-linked list for all the bundles intersecting the sweepline (bundles alternate colors), supporting insertion, deletion of the bundles, and sequential search, 
 and (3) two balanced binary trees (one per color) 
storing all the red and blue bundles in order, and supporting splitting and merging of bundles.
 
At the event point $p_i$
the algorithm processes 
the intersections witnessed by $p_i$, updates the sweepline from $\ell_{i-1}$ to $\ell_i$, 
and makes the necessary changes to bundles (i.e., splits or merges them). 
By proceeding this way, the algorithm maintains the invariant  
that all the red-blue intersections whose witness 
is to the left of the current event point $p_i$ 
are already encountered. 
We summarize the result in the following. 

\begin{theorem}[\cite{ms01}]
The Mantler-Snoeyink algorithm runs in $O(n\log{n})$ time, requires $O(n)$ space, and encounters $O(n)$ bundle-bundle intersections in total.
\end{theorem}

\subsection{Building the life table}
\label{sec:ds}

In this section we describe our algorithm to build the life table for the sets $R$ and $B$. 
Figure~\ref{fig:exe} illustrates the execution of the algorithm for a simple example. 

Before we start our description,  recall~\cite{dss86} that every pointer-based
data structure
with constant in-degree can be transformed into a partially persistent one.
Such a persistent data structure allows accessing in constant
time the data structure at any moment in the past, and performing
pointer operations
on it (but not modifying it); the total time and space required to
make a data structure
partially persistent is linear in the number of structural changes it underwent.

To build the life table for $R$ and $B$, we first perform the Mantler-Snoeyink plane sweep algorithm (see Section~\ref{sec:mantler-snoyeink}), 
 making the  sweepline status structure {partially persistent}.
This ensures that each blue bundle that has 
appeared during the algorithm, 
can afterwards be retrieved from the version of the sweepline status 
at the corresponding 
moment in the past. In particular, 
we are interested in the blue bundles that intersect red bundles. 
We assign each such blue bundle $B_i$ 
a \emph{timestamp} $t_i$ 
reflecting the moment when the first bundle-bundle intersection involving $B_i$ was witnessed.
In order to distinguish between two different bundle-bundle intersections 
discovered at the same moment $t_k$ (i.e., witnessed by the same point), 
we assign the moment 
 $t_k+\epsilon$ to the intersection
that has  smaller $y$ coordinate. 
Figure~\ref{fig:exe}, right, lists all such  blue bundles for the given example.   

Observe that the plane sweep algorithm induces a partial order among the red segments:   
At any moment, the red segments crossed by the sweepline can be ordered from bottom to top.
Since the  red segments are pairwise non-intersecting, 
no two segments may swap their relative position.  
Let $r_1, \ldots, r_n$ be a total order   consistent with the  partial order along the sweepline at each moment.
In Figure~\ref{fig:exe}, the red segments are named according to such an order.

We now build the \emph{life table} of red segments and blue bundles,  
see Figure~\ref{fig:exe}, bottom. 
The life table is a graph defined as follows. 
On its $y$ axis it has integers from $0$ to $n_R$, where $n_R$ is the number of red 
segments;
the  $x$ axis of the life table coincides with the $x$ axis of the original setting, i.e., of the plane $\mathbb{R}^2$.   
Each red segment $r_i$ is represented by a horizontal line segment whose $y$ coordinate  equals $i$  
and whose endpoints' $x$ coordinates coincide with the $x$ coordinates of  
the endpoints of $r_i$. 
Each blue bundle $B_j$, that has  participated in at least one bundle-bundle intersection, 
is retrieved from the version of the sweepline status at the moment $t_j$
 when the first such intersection has been witnessed; $t_j$ is the timestamp of $B_j$. 
In the table, 
$B_j$ is represented by a vertical 
line segment (that could possibly be a point), whose $x$ coordinate is $t_j$.  
This vertical segment intersects 
exactly  the segments representing all red segments intersected by 
 bundle $B_j$ (i.e., the segments of the red bundle(s) participating in the bundle-bundle intersection(s) with $B_j$). 
In particular, the bottom and the top  endpoints of this segment lie  respectively on the two red segments that represent 
 the  first and the last segment in $R$ intersected by bundle $B_j$, 
according to the topological ordering of the red  segments. 
If $B_j$ intersects only one red segment, then in the life table $B_j$ is represented by a point.
In Figure~\ref{fig:exe} all the blue bundles except $B_5$ are represented by a point, but  
in a more complicated example many bundles might
 be represented by  line segments. Note that 
instead of storing the segment list of each blue bundle explicitly, 
we just maintain a pointer to that bundle as it appears  
in (the corresponding version of) the sweepline status structure.

\begin{figure}[h]
\centering
\includegraphics{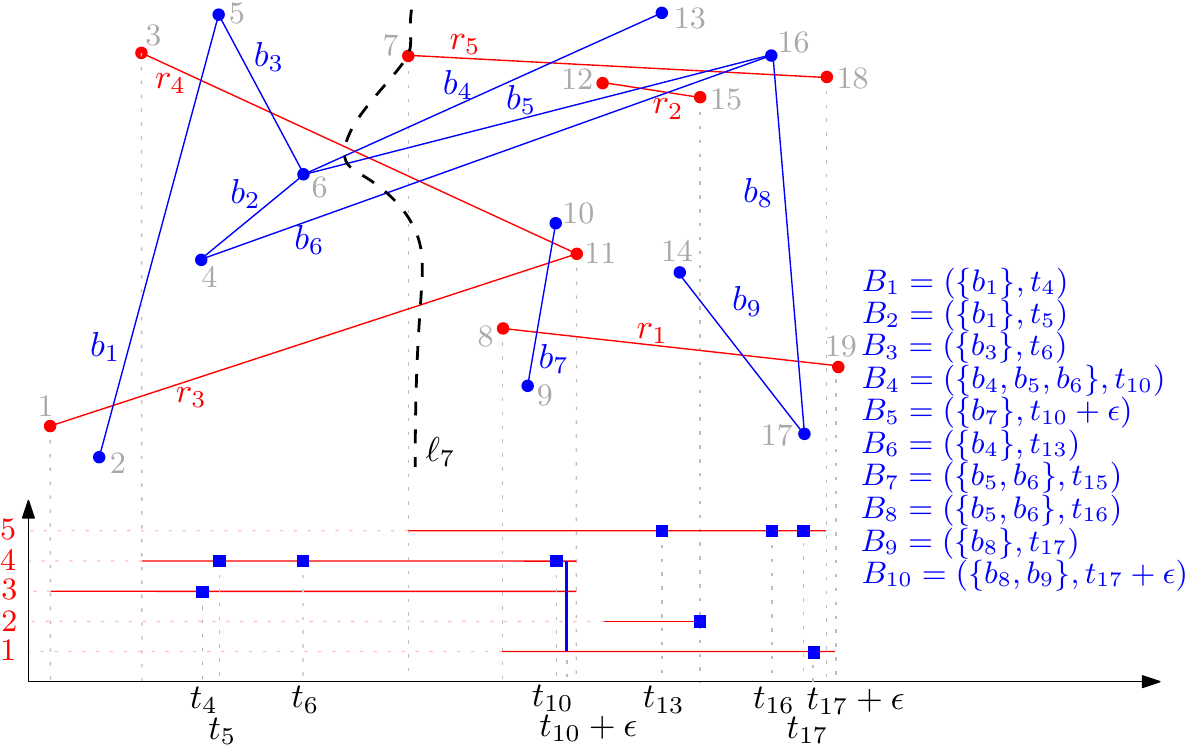}
\caption{Execution  of the algorithm from Section~\ref{sec:ds} 
for $R = \{r_1,\ldots,r_5\}$ and $B = \{b_1,\ldots,b_{10}\}$. Above: the segments in $R$ (red) and $B$ (blue); 
the events of the plane sweep in order (gray numbers), and the pseudoline $\ell_7$ at time $7$ (dashed line).
Right: the set of blue bundles $\{B_1,\ldots,B_{10}\}$  encountered by the algorithm.
Below: the life table for $R$ and $\{B_1,\ldots,B_{10}\}$.}
\label{fig:exe}
\end{figure}

\subsection{The resulting  data structure}
\label{sec:resulting-ds}

After the life table is built, we sweep it with a horizontal straight line from bottom  to top, 
again making the sweepline status partially persistent. 
The events now correspond to red segments, and the version of the 
sweepline status at a time moment $i$ contains all blue bundles crossing the horizontal line $y = i$, 
sorted by $x$ coordinate and stored in a balanced binary tree. 

Our ultimate data structure is the persistent sweepline status of the above (second) plane sweep.  
We are  able, in $O(1)$ time,  to retrieve the version of the sweepline status structure at any moment $i$. 
The sweepline status at the moment $i$ is a tree storing the  blue bundles
whose beginning was witnessed before the moment $i$ and whose end was witnessed after that moment 
(in other words, those blue bundles that intersect the horizontal line $y = i$ in the life table). See Figure~\ref{fig:exe}. 
Since each single bundle is stored in a balanced binary tree, the tree of bundles is also a balanced binary tree. 
Therefore it has height $O(\log{n})$. Moreover,  it can be accessed  in the same way as a standard binary tree: 
any navigation query in it (see Problem~\ref{prob:rb}) can be performed in  $O(1)$ time. 

Now suppose we wish to perform a 
search on (a portion of) a red segment $r_i$ among the blue segments that cross it. 
We are required
to be able to quickly determine  where the interesting intersection(s) lie with respect to $p$,  
for any point $p$ of intersection  between $r_i$ and a blue segment. 
The search in our data structure is then done as follows. 
We retrieve the version of the sweepline status (of the second plane sweep) at the  moment $i$. 
This sweepline status is an implicit balanced binary tree, as explained above. 
We locate the endpoints of  $r_i$ in that tree. Then we 
search in the portion of the tree between $r_i$'s endpoints. 
The decisions during the search are made based on our knowledge about the interesting intersections.

We conclude with the following.  

\begin{theorem}
\label{thm:rb-search}
Given a pair $R, B$ of plane straight-line graphs with $n$ 
edges and vertices in total in both graphs, 
the \theProblem { } for $R$ and $B$ can be solved
 in $O(n\log{n})$ time and $O(n)$ space, 
such that 
the resulting data structure answers each of the {navigation queries} in  $O(1)$ time. 
\end{theorem}

\begin{proof}
The first phase of our procedure to build the data structure is an execution of the Mantler-Snoeyink algorithm, 
with the only difference that the bundle trees from the sweepline status are made partially persistent. 
The latter can be performed with amortized  $O(1)$ time and space overhead
 per update step and a worst-case time cost of $O(1)$ per access step~\cite{dss86}. The total 
number of updates made to the sweepline status 
during the course of the Mantler-Snoeyink algorithm is $O(n)$~\cite{ms01}.
Thus, after the first phase is completed, we have  the order of red segments and the persistent sweepline status. 
With this information, the life table can be built in $O(n)$  time and space: 
we fill the table with the horizontal red segments, and we access sequentially all 
the versions of the sweepline, retrieving the blue bundles and the information on their intersections with the red bundles, and drawing the vertical segments of the life table. 
 Sweeping the life table with a horizontal line, and making the sweepline status partially persistent 
 again costs $O(n\log{n})$ time and $O(n)$ space.

For an edge $e=r_i$ of $R$, the version of the  sweepline status structure at time $i$  provides a balanced binary search tree 
 $T_e$, required by the \theProblem, such that $T_e$ can be navigated (but not modified) in the same way and with the same time complexity as the 
standard balanced binary search tree. Hence the \emph{navigation queries} of the \theProblem { } can be answered in $O(1)$ each. \end{proof} 

\section{Applications}
\label{sec:appl}

We proceed with more detail on the applications of our technique, which are listed in Section~\ref{sec:intro}.

\subsection{The red-blue batched search problem}
\label{sec:batched}

Consider two sets, $R$ (red) and $B$ (blue),  of line segments in $\mathbb{R}^2$,  
such that the  segments in each set are pairwise interior-disjoint, and suppose 
that some of the red-blue intersections are \emph{interesting}, and there is at most one interesting red-blue intersection per each segment. 
Let $\oracle$ be an oracle that, given an intersection point $p$ of a red segment $r$ and a blue segment $b$, determines
 to which side of $p$ the {interesting intersection} on $r$ lies. 

\begin{problem}[Red-blue batched search problem~\cite{dmt06}]
\label{prob:batched}
Given sets $R, B$ and oracle $\oracle$,
find all interesting intersections between the  segments in $R$ and the ones in $B$. 
\end{problem}

Dehne et~al.~\cite{dmt06} showed how to solve the red-blue batched search problem 
by using an augmentation of  the hereditary segment tree data structure of Chazelle et al.~\cite{CEGS94}. 
Their solution requires $O(n\log^2{n})$ space and  $O(n\log^3{n})$  
queries to the oracle. 

Our technique presented in Section~\ref{sec:ds} can be directly applied to solve the red-blue batched search problem with better time and space:
We preprocess the sets $R$ and $B$; after that for each red segment $r$ we perform a binary search in the (implicit) tree storing the red-blue intersections along $r$. 
The search is guided by the oracle $\oracle$, and thus it requires $O(\log{n})$ queries to the oracle. 
Since the number of red segments is $O(n)$, the total number of queries to the oracle 
required for searching all red edges is $O(n\log{n})$. Theorem~\ref{thm:rb-search} implies the following. 

\begin{theorem}
\label{thm:batched}
The red-blue batched search problem for the sets $R,B$ and the oracle $\oracle$
 can be solved using $O(n)$ space and $O(n\log{n})$ queries to the oracle, where $n$ is the total number of segments in $R$ and $B$. 
\end{theorem}

\subsection{Vertical distance for a pair of 3D polyhedral terrains, one of which is convex}
\label{sec:distance}

Let $R$ and $B$ be two polyhedral terrains of complexity $n_R$ and $n_B$ respectively, where terrain $B$ is convex 
(that is, $B$ is the upper envelope of a set of $n_B$ planes in $\mathbb{R}^3$). Let $n = n_R + n_B$. 
We wish to determine the maximum vertical distance between $R$ and $B$, i.e., the length of the longest 
vertical line segment connecting a point in $R$ and a point in $B$.\footnote{It may happen, that the distance keeps increasing as we move along
some direction towards infinity. Then we say that the maximum vertical distance between $R$ and $B$ is $+\infty$.}
As an illustration, the reader may imagine a (convex) approximate model of a mountain, and a need to compare it with 
the real mountain (of course, not necessarily convex) in order to estimate the quality of the approximation. 

Since  both surfaces $R$ and $B$ are composed of planar patches, the vertical distance between $R$ and $B$
is the vertical distance between these patches. 
The maximum vertical distance between $R$ and $B$ is thus attained either between a vertex of 
$R$ and a facet of $B$ (or vice versa), at infinity along an unbounded edge of one of the surfaces, or between an edge  of $R$ and an edge of $B$.
In the second case the distance between $R$ and $B$ is $+\infty$. Both the first and the second case 
can be easily processed by point location 
queries of a (possibly infinite) vertex of one surface into the other one, which requires $O(n\log{n})$ time in total. 
To deal with the last case one can preprocess the vertical projections of $R$ and $B$ following our technique, and perform the binary searches along each edge $e$ of $R$ for the intersection with an edge of $B$ maximizing the vertical distance.  
Consider the cross-section of $B$ by the vertical plane containing $e$. This is a convex monotone polygonal line. 
The sequence $h_B$ of heights 
of its breakpoints is unimodal, and since all points of $e$ lie on the same line, if we subtract  from each member of $h_B$ the 
height of the point in $e$ lying on the same vertical line, the resulting sequence will still have one maximum. It then follows that given a point $p \in e$ vertically above/below an edge of $B$, in constant time we can find out in which direction this maximum lies, and this is exactly 
what the oracle for the binary search  along $e$ should do. Using Theorem~\ref{thm:rb-search}, we conclude. 

\begin{theorem}
\label{thm:vert}
Given a pair 
of polyhedral terrains in 3D,  
where one of the terrains is convex, 
the maximum vertical distance between the terrains can be found 
in $O(n\log{n})$ time and $O(n)$ space, where $n$ is the total complexity of both terrains.  
\end{theorem}

Notice that by slightly changing the algorithm, we could be answering the minimum vertical distance, instead of the maximum one. 
In particular this gives an alternative $O(n\log{n})$ algorithm to solve the shortest watchtower problem~\cite{S88,Z97}.

\subsection{Construction of the Hausdorff Voronoi diagram}
\label{sec:hvd}
Given a set of  $n$ distinct points in the plane, we partition this set,  
resulting in a family  $S$ 
of point clusters, where no two clusters share a point.
Let the distance from a point $t \in \mathbb{R}^2$ to a cluster $P \in S$, denoted as $\df{t,P}$,
 be the maximum Euclidean distance
from $t$ to any point in  $P$. 
 The \emph{Hausdorff Voronoi diagram} of $S$, denoted as $\HVD{S}$,
is a subdivision of $\mathbb{R}^2$ into maximal regions such that every point within 
one region has the same nearest cluster according to distance $\df{\cdot,\cdot}$.  

The diagram has  
worst-case combinatorial complexity $\Theta(n^2)$, and it can be constructed in optimal $O(n^2)$ time~\cite{EGS89}. 
However, these bounds can be refined according to certain parameters of the family  $S$. 
Two clusters are called \emph{non-crossing} if their convex hulls
intersect at most twice (i.e., their convex hulls are pseudocircles), and \emph{crossing} otherwise.
Below we consider separately the (simpler) case of non-crossing clusters, and the one of crossing clusters. The latter case subsumes  the former one. 

\paragraph{Non-crossing clusters.}  
If all clusters in $S$ are pairwise non-crossing, the complexity of $\HVD{S}$ is $O(n)$. 
In this case the diagram can be constructed in expected $O(n\log^2{n})$ time and expected $O(n)$ space~\cite{CKLP16,kp16a}. 
The best deterministic algorithm to date requires  
$O(n\log^5{n})$ time and $O(n\log^2{n})$ space~\cite{dmt06}.\footnote{The time complexity claimed in Dehne~et al. is $O(n\log^4{n})$, however we 
believe that in reality the described algorithm requires $O(n\log^5{n})$ time.  
} 
The latter algorithm follows the divide-and-conquer strategy. 
To merge two recursively computed diagrams, a bottleneck procedure is formulated as
a red-blue batched segment search problem (see Section~\ref{sec:batched}), 
where the two segment sets are the sets of edges of the two diagrams. 
The latter problem is then solved in $O(n\log^2{n})$ space and $O(n\log^3 n)$ queries to the oracle.
The 
authors define an oracle to perform this search, which they assume can be implemented in $O(1)$ time. 
We were unable to reconstruct the claimed constant-time oracle, 
however we know how to implement it in $O(\log{n})$ time per query. 
Theorem~\ref{thm:batched} implies an algorithm to construct the Hausdorff Voronoi diagram of a family of non-crossing clusters 
in $O(n\log^3n)$ time and $O(n)$ space.\footnote{Note that if it was possible to implement the oracle in $O(1)$ time, 
our algorithm would instantly be improved by a $O(\log{n})$ factor. Thus in all cases 
our algorithm is faster than the previous one by a factor of $O(\log^2 n)$.} This result is 
subsumed by  
the one for arbitrary clusters, 
which is strictly more general 
(for a family of non-crossing clusters $m=0$).

\paragraph{Arbitrary clusters.}
Consider the Hausdorff Voronoi diagram of a family $S$ of  arbitrary clusters, that could possibly cross. 
The essential parameter used to refine the quadratic bounds related to the diagram in that case, is the \emph{number of crossings},\footnote{See~\cite{p04,pl04}
for the formal definition of the number of crossings.} 
denoted by $m$. The parameter $m$ is bounded from above by half
 the number of intersections  between the convex hulls of all pairs of crossing clusters.
  In the worst case $m = \Theta(n^2)$, 
however it is  small in known  practical 
applications, e.g., in VLSI CAD~\cite{p04,pl04}.
The combinatorial complexity of the Hausdorff Voronoi diagram is shown to be $O(n+m)$~\cite{p04}. 
Apart from the $O(n^2)$ time algorithm mentioned above, there is a plane sweep~\cite{p04} and a divide-and-conquer~\cite{pl04} 
algorithm to construct the diagram. 
Both of them are sensitive to the parameter $m$, however their time complexity depends as well on some other parameters, which are
unrelated to $m$. 
In particular, these algorithms may have $\Omega(n^2)$  time complexity even when clusters are non-crossing. 

Our technique can be applied to reduce the
time complexity of the divide-and-conquer construction of 
the Hausdorff Voronoi diagram of arbitrary clusters~\cite{pl04}. The resulting algorithm is the fastest to date deterministic algorithm for certain
input  families, where clusters may cross, but the number of crossings is small. 
Since the description of our algorithm requires a lot of additional definitions and details, we defer it to Section~\ref{sec:hvd-arb}.

\subsection{Construction of the farthest-color Voronoi diagram}
\label{sec:fcvd}
Let $S$ be a family of clusters of points in the plane, induced by partitioning a set of $n$ points in the plane -- the same setting 
as in the previous section. 
Here we consider the \emph{farthest-color Voronoi diagram}~\cite{hks1993dcg,ahiklmps01} of $S$, a 
generalized cluster Voronoi diagram, which is in some sense the opposite to the Hausdorff Voronoi diagram. 
It is defined as a subdivision of the plane into maximal regions, 
such that for any point in the region $\fcreg{P}$, $P \in S$,  
the  cluster $P$ is the farthest cluster in $S$. 
The distance from a point $t \in \mathbb{R}^2$ to a cluster $P \in S$ is the minimum distance from $t$ to a point in $P$. 
In  $\FCVD{S}$, 
each $\fcreg{P}$  is subdivided into finer regions of the points of $P$ 
 by the nearest-neighbor Voronoi diagram 
 of the points of that cluster, $\VD{P}$. 
The edges of this additional subdivision are called \emph{internal edges}. 

The farthest-color Voronoi diagram is much less understood than the Hausdorff diagram. 
Its combinatorial complexity is $O(nk)$, and a matching lower bound is known for $k \leq n/2$~\cite{ahiklmps01}, where $k$ is the number of clusters in $S$. 
An $O(n^2)$ time construction algorithm is implied by the result of Edelsbrunner et al.~\cite{EGS89}. 
Since the diagram is not an instance of \emph{farthest abstract Voronoi diagrams}~\cite{M01}, studying 
particular families of clusters (proving better bounds on the complexity of the diagram, and finding faster construction algorithms)
is a non-trivial task.
Recently it was shown that if the clusters in $S$ are 
pairs of endpoints of disjoint segments parallel to each other, 
$\FCVD{S}$ has $O(n)$ complexity and can be computed in optimal $O(n\log{n})$ time~\cite{CKPSS17}.

Applying our technique, we are able to claim near-optimal time complexity for a larger class of input families, where each 
cluster in a family is a quadruple of points that are corners of an axis-aligned rectangle, and the corresponding rectangles 
are pairwise disjoint.
We first need to show that the farthest-color Voronoi diagram of such an input has linear combinatorial complexity, see Proposition~\ref{prop:fcvd-compl} below.
To prove this, we make use of the following.

\begin{figure}
\centering
\includegraphics{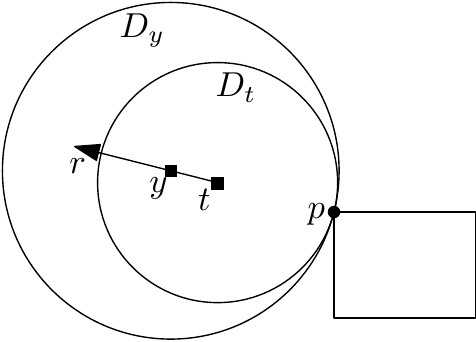}
\caption{Illustration for the proof of Lemma~\ref{lemma:fcvd-bounded}}
\label{fig:disks}
\end{figure}

\begin{lemma}
\label{lemma:fcvd-bounded} 
Let $f$ be a bounded face of $\FCVD{S}$. 
Then $f$ is incident to an internal edge of $\FCVD{S}$. 
\end{lemma}

\begin{proof} 
The following proof is a simple generalization of one of the  arguments used to prove~\cite[Lemma~3]{CKPSS17}. 
Face $f$ is a face of $\fcreg{p}$ for some point $p\in P, P \in S$. 
Let $t$ be a point in $f$, and  let $r$ be the ray originating at $t$ and having direction $\overrightarrow{pt}$. 
See Figure~\ref{fig:disks}. 
Consider the closed disk $D_y$ centered at a point $y \in r$ and passing through  $p$. 
Since $t \in \fcreg{p}$, $D_t$ intersects all the clusters in $S$. 
As $y$ moves along $r$, the disk grows and it still intersects all clusters in $S$. 
Until $y$ hits an edge of $\VD{P}$, the minimum distance from $y$ to $P$  is attained by $p$, and since $D_y$ intersects 
all clusters in $S$, $y$ lies in $\fcreg{p}$. Thus either the whole ray $r$ is contained in $f$, or the first 
 intersection between $r$ and the boundary of $f$ lies on an internal edge of $\FCVD{S}$. Since we assumed $f$ to be bounded, 
the latter must hold. 
 \end{proof}

\begin{proposition}
\label{prop:fcvd-compl}
For a family $S$ of $n$ clusters, 
where each cluster is all the corners of an axis-aligned rectangle, and these rectangles are pairwise disjoint,
the combinatorial complexity of  
$\FCVD{S}$ is $O(n)$. 
\end{proposition}

\begin{proof}
$\FCVD{S}$ is a plane graph, whose  each  vertex has degree at least three. 
Thus by Euler's formula it is enough to show that the number of faces of $\FCVD{S}$ is $O(n)$. 
We will treat separately its bounded and unbounded faces. 
We assume that each cluster has 
four 
 distinct points, as otherwise (when it has two or one point) the situation becomes simpler.

We first  show that the total number of  bounded faces of $\FCVD{S}$ is $O(n)$.
By Lemma~\ref{lemma:fcvd-bounded}, any bounded face $f$ of  $\FCVD{S}$
is incident to  an internal edge of $\FCVD{S}$. 
Recall that each internal edge of $\FCVD{S}$ 
is a portion of an edge of $\VD{P}$,
 that lies in $\fcreg{P}$,  for some cluster $P\in S$.
In other words, an edge of $\VD{P}$ intersects $\fcreg{P}$ in several connected components (line segments or rays), 
and each such connected component is an internal edge of $\FCVD{S}$. 
 In the next paragraph,  we will show that each edge of  $\VD{P}$ of any cluster $P \in S$ contributes at most two internal edges to $\FCVD{S}$.
This will imply that  the total number of internal edges in $\FCVD{S}$ is $O(n)$: 
Since each cluster $P \in S$ is a quadruple of points that are the corners of an axis-aligned rectangle, 
$\VD{P}$ has four edges 
(see Figure~\ref{fig:fcvd}a), and therefore 
the total number of edges in the nearest-neighbor Voronoi diagrams of all clusters in $S$ is $O(n)$.

Let $e$ be an edge of $\VD{P}$, 
and let $p,p' \in P$ be the points  that induce $e$, i.e., $e$ is a 
ray 
contained in the bisector of $p,p'$. See Figure~\ref{fig:fcvd}b.
The 
segment $pp'$ breaks $e$ into two portions; consider one of them, $e_r$; the other part is treated analogously. 
Suppose for the sake of contradiction that 
$e_r$ intersects 
 $\fcreg{P}$ in at least two connected components. 
Consider the circle passing through points $p,p'$ and whose center $y$ moves along $e_r$ in the direction of growing radius of the circle. 
When $y$ stops being in  $\fcreg{P}$, the circle stops containing a point of every cluster in $S$. That is, some cluster $Q$ starts to be outside 
the circle. Since the circle is growing, there is no way for $Q$ to restart intersecting the circle, unless the convex hull of $Q$ intersects the segment $pp'$ (see the dark-blue rectangles in Figure~\ref{fig:fcvd}b), 
which would contradict the disjointness of the convex hulls of the clusters in $S$. 

Now we estimate the number of unbounded faces of $\FCVD{S}$.
Consider the arrangement of all edges of the nearest-neighbor Voronoi diagrams of clusters in $S$. This is a grid made of $n$ vertical and $n$ horizontal lines.  
Within one cell of that grid, $\FCVD{S}$ coincides with the farthest Voronoi diagram of $n$ points (one point per cluster). 
The number of unbounded regions of $\FCVD{S}$ equals the number of unbounded edges of $\FCVD{S}$. 
That number is comprised of the total number of unbounded internal edges (it is $O(n)$ since each of them is a portion of a line forming the grid)
and the total number  of the unbounded edges of $\FCVD{S}$ within the unbounded cells of the grid. 
Now we observe that such cells are 
$O(n)$ half-strips and four quarter-planes. 
At infinity, the total number of distinct directions that all the half-strips correspond to 
is only four, thus 
at most four additional unbounded edges of $\FCVD{S}$ may lie within these strips. 
At each quarter-plane, we have the farthest-point Voronoi diagram of $n$ points, which has $O(n)$ unbounded edges. 
Therefore the total number of unbounded edges of $\FCVD{S}$ is $O(n)$, and the same bound holds for the total number of unbounded faces of $\FCVD{S}$. 
\end{proof}

\begin{figure}
\begin{minipage}{0.49\textwidth}
\centering
\includegraphics[scale=0.7]{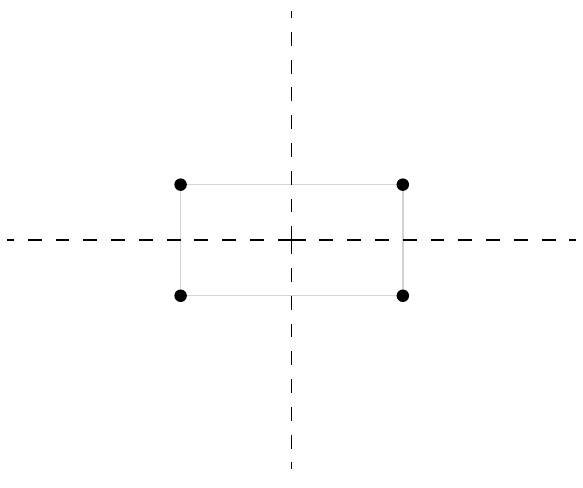} \\
(a)
\end{minipage}
\begin{minipage}{0.49\textwidth}
\centering
\includegraphics{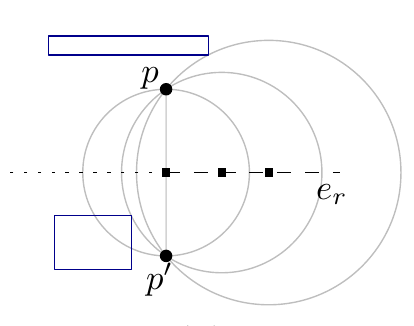} \\
(b)
\end{minipage}
\caption{(a) A cluster $P$, its convex hull (gray), and its nearest-neighbor Voronoi diagram $\VD{P}$ (dashed); (b) an illustration for the proof of Proposition~\ref{prop:fcvd-compl}}
\label{fig:fcvd}
\end{figure}

We proceed with the main  result of this section.

\begin{theorem}
For a family  $S$ of $n$ clusters, that are corners of pairwise disjoint axis-aligned rectangles,
$\FCVD{S}$  has complexity $O(n)$ and it can be constructed in $O(n\log^2{n})$ time. 
\end{theorem}

\begin{proof}
The complexity of $\FCVD{S}$ is $O(n)$ due to the above Proposition~\ref{prop:fcvd-compl}. 

To construct $\FCVD{S}$, we employ the divide-and-conquer strategy. 
We are not able to 
apply the divide-and-conquer algorithm from~\cite[Lemma~24]{CKPSS17}, since the merge curve is not necessarily connected, and may have cycles. 
Suppose that for two subsets $S_\ell$, $S_r$ of $S$, the farthest-color Voronoi diagram is already computed. 
To merge two diagrams, we preprocess them using our technique from Section~\ref{sec:ds}, and obtain the implicit search trees for the edges of 
$\FCVD{S_r}$.
Then we can search each internal edge of $\FCVD{S_r}$, to determine all its portions that appear in the merged diagram.
If we subdivide the edge by the segment through its generator points, the two halves will have at most one such connected component each, 
see the proof of Proposition~\ref{prop:fcvd-compl}. 
Thus we search them separately. To navigate the search along such a portion $e_r$, given a point $p$ of intersection  between $e_r$ and an edge $e'$ of $\FCVD{S_\ell}$,
we obtain the cluster $Q$
that is farthest from $p$ among the clusters in $S_\ell$. We construct the bisector between $P$ and $Q$, and this shows to which side of $p$
there lies the portion of $e$ that is farther from $P$ that from $Q$. That portion is the only one to which we should continue. 

The merging step overall takes $O(n\log n)$ time and $O(n)$ space  to preprocess the diagrams, and $O(\log{n})$-time search per each of $O(n)$ 
internal edges 
 of $\FCVD{S_\ell}$. 
The claimed time and space complexity follows.
\end{proof}

\subsection{The stabbing circle problem} 
\label{sec:stabbing}
Given a set $S$ of $n$ line segments in the plane, a circle $c$ is called a \emph{stabbing circle} for $S$ if 
every segment in $S$ has exactly one endpoint in the exterior of the disk induced by $c$.
Two stabbing circles $c_1,c_2$ are \emph{combinatorially different} if they classify the endpoints of $S$ differently.

The \emph{stabbing circle problem} for $S$  
consists of computing  (a representation of) all the combinatorially different stabbing circles for $S$ (if they exist); 
and finding stabbing circles with the minimum and maximum radius. 
Although the stabbing circle problem can be solved in  a worst-case optimal $O(n^2)$ time and space 
by applying a technique by Edelsbrunner et al.~\cite{EGS89} (see~\cite{CKPSS17} for the explanation), 
the problem can be solved much faster if the input set of segments is of some particular form~\cite{CKPSS17}. 
The  method to do so is based on 
the Hausdorff and the farthest-color Voronoi diagram (see Section~\ref{sec:hvd} and ~\ref{sec:fcvd} for the definitions), and 
its time complexity depends on the parameters of these diagrams and of the input segment set; we will detail it in the next paragraph. 
For the segments parallel to each other, it works in $O(n\log^2 n)$ time. 
 The  technique presented in Section~\ref{sec:ds} helps to improve the time complexity of this alternative method by a
 $O(\log n)$ factor. 
This automatically reduces the time required to solve the stabbing circle 
problem for parallel segments from $O(n\log^2{n})$ to $O(n\log{n})$, 
which now  matches the lower bound given in~\cite{ckpss16-eurocg}.

Let $\HVD{S}$ and $\FCVD{S}$ denote respectively the Hausdorff and the farthest-color Voronoi diagram, 
whose sites are pairs of endpoints of the segments in $S$. 
Let $\mathcal{T}_{\HVD{S}}$ and $\mathcal{T}_{\FCVD{S}}$ denote the time required to compute these two diagrams, 
and let $|\HVD{S}|$ and $|\FCVD{S}|$ respectively denote their combinatorial complexity. Additionally, we let $m$
be a parameter reflecting interaction between the segments in $S$.\footnote{ 
$m$ denotes
 the number of pairs formed by a segment $aa' \in S$ and a pure edge $uv$ of $\HVD{S}$ 
such that 
$a \in D(u) \setminus D(v)$ 
and $a' \in D(v) \setminus D(u)$. $D(u)$ (resp., $D(v)$) is the  disk
centered at $u$ (resp., at $v$) with the radius $\df{u,P}$ (resp., $\df{v,P}$), where $P$ is the cluster in $S$ 
such that $e$ is incident to $\hreg{P}$.} 
We are ready to state the new result 
for the stabbing circle problem for $S$ that improves on~\cite[Lemma~17; Corollary~3, Theorem~4]{CKPSS17}. 

\begin{theorem}
\label{thm:stabbing}
Given a set $S$ of $n$ segments in the plane in general position,
the stabbing circle problem for $S$ can be solved in time  
$O(\mathcal{T}_{\HVD{S}} +\mathcal{T}_{\FCVD{S}}+(|\HVD{S}|+|\FCVD{S}|+m)\log n)$. 
If the segments in $S$ are parallel to each other, the stabbing circle problem for $S$
can be solved in optimal $O(n\log{n})$ time and $O(n)$ space.  
\end{theorem}

\begin{proof}[sketch]
The only substantial change to the algorithm in~\cite{CKPSS17} is an alternative procedure for the \emph{find-change query}. 
In particular,~\cite[Lemma~17]{CKPSS17} gives a procedure to perform this query in $O(\log^2{n})$ time by a nested point location in $\FCVD{S}$. 
By using the technique presented in Section~3, we replace the nested point location by a binary search on $e$
 among the intersections with the edges of $\FCVD{S}$.
The oracle for our search works in $O(1)$ time: 
Given a point $p$ on the intersection of an edge $e$ of $\HVD{S}$ and $g$ of $\FCVD{S}$, the ``owners'' of $p$ in both diagrams are available automatically. 
Once the owners of $p$ are known, determining the \emph{type} of $p$ is a constant-time operation. 
Using Theorem~\ref{thm:rb-search},  we improve the result of~\cite[Lemma~17]{CKPSS17} by a factor of $O(\log{n})$.
The claim follows 
by plugging the improved result into~\cite[Theorem~1]{CKPSS17}.
\end{proof}

\section{A Divide and Conquer Algorithm for the Hausdorff Voronoi Diagram}
\label{sec:hvd-arb}

Let $S$ be a family of
point clusters in the plane, such that no two clusters have a common point. Let $n$ be the total number of points in $S$.  
We follow a general position assumption  that  
no four points  lie on the same circle.
We also assume that no cluster encloses another in its convex hull, as
the outer  cluster would have empty region in the Hausdorff Voronoi diagram.

\begin{definition}
Two clusters $P$ and $Q$ are called \emph{non-crossing},  
if the convex hull of $P\cup Q$ admits at most two supporting segments with one endpoint in $P$
and one endpoint in $Q$. 
If the convex hull of $P\cup Q$ admits more than two such supporting segments, then $P$
and $Q$ are called 
\emph{crossing}.	
\end{definition}

The \emph{farthest Voronoi diagram} of a cluster $C$, denoted as $\FVD{C}$, 
is a partitioning of the plane into regions 
where the 
\emph{farthest Voronoi region} of a point $c \in C$~is 
$
\freg{C}{c}  = \{p \mid
\forall c' \in C \setminus\{c\} \colon
d(p,c) > d(p,c') \}. $
The graph structure of $\FVD{C}$ is a tree $\fskel{C} =  \mathbb{R}^2
\setminus \bigcup_{c \in C} \freg{C}{c}$. 

Given a cluster $C \in S$ and a point $p \in \mathbb{R}^2$, we let $\df{p,C}$ denote the \emph{maximum distance} 
between $p$ and the points of the cluster $C$, i.e., $\df{p,C} = \max_{c \in C}{d(p,c)}$, 
where $d(\cdot,\cdot)$ denotes the Euclidean distance between two points. 

The \emph{Hausdorff Voronoi diagram} of $S$, denoted as $\HVD{S}$, is a partitioning of the plane into 
regions, where the  \emph{Hausdorff Voronoi region} of 
a cluster $C \in S$ is   
$ 
\hregs{S}{C}  = \{p \mid \forall C'\in S\setminus\{C\} \colon \df{p,C}
< \df{p,C'} \}.
$ 
The \emph{Hausdorff Voronoi region} of a point $c \in C$ is 
$\hregs{S}{c}  = \hregs{S}{C} \cap \freg{C}{c}.$ 
If clear from the context, we do not write subscripts for $\hrego{}$ and $\frego{}$.

For two 
clusters $P,Q \in S$, their Hausdorff bisector $\bh{P,Q} = \{y \mid  \df{y,P} = \df{y,Q} \}$ 
consists of one (if $P,Q$
are non-crossing) or more (if $P,Q$ are crossing) unbounded polygonal chains~\cite{p04,pl04}. Each   vertex 
of $\bh{P,Q}$ is the center of a circle passing through two points of
one cluster and one point of another that entirely  encloses $P$ and
$Q$. 
The vertices of $\bh{P,Q}$ are called \emph{mixed} vertices. 

\begin{definition}
\label{def:crossing}
A  mixed vertex on the bisector $\bh{P,Q}$, induced by two points $p_i, p_j \in P$ and a point $q_l \in Q$  
is called \emph{crossing}, if there is a diagonal $q_lq_r$ of $Q$
that crosses the diagonal $p_ip_j$ of $P$, and  all points $p_i,p_j,q_l,q_r$ are on the convex hull of $P\cup Q$. 
The total number of crossing vertices along the bisectors of all pairs of clusters is the \emph{number of crossings}
and is denoted by $m$.
\end{definition}


Edges of the Hausdorff Voronoi diagram are of two types: \emph{pure} edges, and  \emph{internal} edges, that are portions of edges of 
$\fskel{C}$ inside $\hreg{C}$. Each pure edge separates the Hausdorff Voronoi regions of two different  clusters, and each internal edge separates 
the Hausdorff Voronoi regions of two different point in one cluster. 
Another way to see it is the following. Any edge of the Hausdorff Voronoi diagram 
is a portion of a bisector between two points. 
If these points are of the same cluster than the edge is a pure edge, otherwise it is 
an internal edge.


\begin{property}[\cite{p04}]
\label{prop:con-comp}   
Each face of a region $\hreg{C}$ 
intersects $\fskel{C}$ in one non-empty
connected component.
\end{property}    

Given a family $S$ of point clusters, 
its Hausdorff Voronoi diagram can be computed in $O(M + n\log^2{n} + (m+K)\log{n})$\footnote{
$m = \sum_{(P,Q)}m(P,Q)$, where $m(P,Q)$ is the number of crossing mixed vertices on the bisector between $P$ and $Q$, 
for any pair of crossing clusters $P,Q$.
$K = \sum_{P\in S}{K(P)}$, $M = \sum_{P \in S}{M(P)}$, where $K(P)$ is the number of clusters 
enclosed in the minimum enclosing circle of $P$, and $M(P)$ is the number of convex hull 
points $q \in Q$ that are interacting with $P$, that is, $q$ is enclosed in the minimum 
enclosing circle of $P$ and either $Q$ is entirely enclosed in the minimum enclosing circle of $P$ or $Q$ is
crossing with $P$.}
using the divide-and-conquer strategy~\cite{pl04}.
We show how  to reduce this to $O((n+m)\log^3{n})$ time and $O(n+m)$ space (Theorem~\ref{thm:hvd-cr}) by using 
the data structure described in Section~\ref{sec:ds}.

Let the input set $S$ be subdivided into two sets $\redClusters$ and 
$\blueClusters$ by a vertical line $\ell$, such that 
the leftmost point of each cluster in $\redClusters$ and in $\blueClusters$ is respectively to the left and to the right of $\ell$. 
Let
the diagrams $\redHVD$ and $\blueHVD$ be recursively computed. 
The  nontrivial part of merging $\redHVD$ and $\blueHVD$ is to find
the merge curve $\mergeCurve$. 
The main complication, that distinguishes this case from the well-known divide-and-conquer algorithm for the Voronoi diagram of points, 
 is that
$\mergeCurve$ is not a Jordan curve, in particular, 
it has several unbounded components and several cycles~\cite{pl04}. 
The unbounded components, as well as the cycles that enclose at least one vertex of $\redHVD$ or of $\blueHVD$, 
can be found efficiently~\cite{pl04}. The cycles that enclose no vertices of $\redHVD$ or $\blueHVD$ must contain a portion of an internal edge of one of the diagrams
(which follows from Property~\ref{prop:con-comp}). Thus to identify them, we need to perform searches along the internal edges of $\redHVD$, and the ones of $\blueHVD$. 
Below we will concentrate on the former case; the latter one is analogous.


We perform these searches by means of the data structure developed in Section~\ref{sec:ds}. 
For any internal edge of $\redHVD$, the data 
structure provides a (implicit) tree storing all the intersections of this edge with the edges of $\blueHVD$. 
We search this tree, that is,  at each accessed node we need to decide which subtree we should continue with. 
Sometimes we should continue to both subtrees, but we are able to charge this branching to 
crossing mixed vertices induced by $p_1$ and $p_2$. 
Each vertex is charged at most once, and their overall number is $m$, which gives us the claimed time complexity. 
We proceed with a more detailed description. 

\begin{figure}
	\centering
\includegraphics{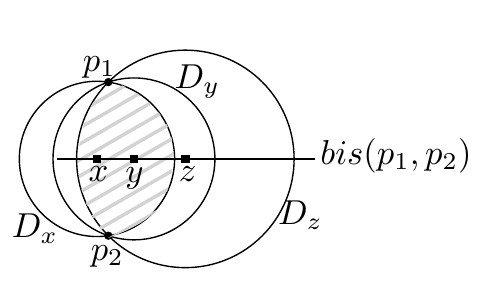}
\caption{$ D_x \cap D_z $ (shaded) is inside $D_y$}
\label{fig:circ}
\end{figure}

The following lemmas are the basis of our decision procedure.
Let $uv$ be a connected portion of an internal edge of $\redHVD$ 
induced by points $p_1, p_2$ of cluster $P \in \redClusters$, 
i.e., $uv$ is a portion of the Euclidean bisector of $p_1$ and $p_2$.  

\begin{lemma}
	\label{lemma:corr}
	For any cluster $Q$  in $\blueClusters$:  
	(i) the locus of points in $uv$ that are closer to $Q$ than to $P$ is connected.
	(ii) the locus of points in $uv$ that are closer to $P$ than to $Q$ may have up to two connected components. 
	If it has two connected components,  
	then there is a pair of crossing mixed vertex induced by $p_1,p_2$ and a point in $Q$. 
\end{lemma}

\begin{proof}	 
	\textbf{(i).}	Suppose for the sake of contradiction 
	that the locus of points in $uv$ that are closer to $Q$ than to $P$  is not connected. 
	Then there are points $x,y,z \in uv$ such that  $y$ is between $x$ and $z$, and 
	$\df{x,Q} < \df{x,P}$,  $\df{x,Q} > \df{x,P}$, and $\df{x,Q} < \df{x,P}$.
        Let $D_x, D_y, D_z$ be the disks centered at respectively  $x,y$ and $z$, 
	and each passing through points $p_1,p_2$. See Figure~\ref{fig:circ}.
	Note that since $uv$ is a portion of an internal edge of $\redHVD$, for each point $q \in uv$, $\df{q,P} = d(q,p_1)$. 
	Thus, $Q \subset D_x$, $Q \not\subset D_y$, and $Q \subset D_z$. 
But $D_x \cap D_z \subset D_y$. We obtain a contradiction.  

	\textbf{(ii).} The first part of the statement follows directly from item  (i) since
 the locus of points in $uv$ closer to $P$ than to $Q$ is the complement of the locus of points considered in (i). 
	Now suppose the locus of points in $uv$ that are closer to $P$ than to $Q$ has two connected components. Thus, for the
	segment $v_1v_2 \subset uv$ in-between these two components, every point in $v_1v_2$ is closer to $Q$ than to $P$.
	Consider the two disks $D_{v_1}$ and $D_{v_2}$ centered respectively at $v_1$ and $v_2$, and each passing through $p_1,p_2$. See Figure~\ref{fig:cr}.
 By construction,  $D_{v_1}$ and $D_{v_2}$ contain respectively  a point $q_1 \in Q$ 
	and a point $q_2 \in Q$ on their boundaries, and these points are to the opposite sides of $p_1p_2$, and both $P$ and $Q$ are contained in each of $D_{v_1}$ and $D_{v_2}$.   
	Thus the chord $q_1q_2$ crosses the chord $p_1p_2$, and all four points appear on the convex hull of $P \cup Q$. 
	Thus $v_1$ and $v_2$ both  are crossing mixed vertices of $\bh{P,Q}$ induced respectively by $p_1,p_2,q_1$ and by $p_1,p_2,q_2$. 
	The claim follows. 
\end{proof}

\begin{figure}
\centering
\includegraphics{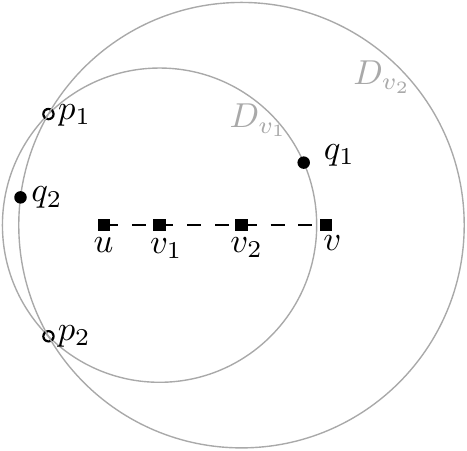}
\caption{Illustration of the proof of Lemma~\ref{lemma:corr}ii}
\label{fig:cr}
\end{figure}

\begin{lemma}
	\label{lemma:oracle}
Let $x$ be a point in $uv$, and let $Q$ be a cluster in $\blueClusters$
such that $\df{x,Q} \leq \df{x,Q'}$ for
	any $Q' \in \blueClusters$.  If $\df{x,Q} < \df{x,P}$, then in $O(\log{n})$ time it is possible to check whether 
        $xv$ contains points closer to $P$ than to $Q$. 
\end{lemma}

\begin{proof}
	We check whether $xv$ contains points closer to $P$ than to $Q$ as follows. 
	We locate point $v$ in $\FVD{Q}$, which gives us $\df{v,Q}$.       
	Recall that $\df{v,P} = d(v,p_1)$, and this is available immediately.
	If $\df{v,Q} \geq \df{v,P}$, we return the positive answer. 
	If $\df{v,Q} < \df{v,P}$, then by Lemma~\ref{lemma:corr}i all the points in $xv$
	are closer to $Q$ than to $P$, and the answer to our query is negative.  
\end{proof}

Now we are ready to state the main result of this section. 

\begin{theorem}
\label{thm:hvd-cr}
Given a family $S$ of point clusters in the plane, 
The Hausdorff Voronoi diagram of $S$ can be 
computed in $O((n +m)\log^3{n})$ time  and $O(m+n)$ space, where 
$n$ is the total number of points in all clusters in $S$, and  
$m$ is the number of crossings for the clusters in $S$. 	
\end{theorem}

\begin{proof}
The algorithm is a divide-and-conquer algorithm, that follows the one 
of Papadopoulou and Lee~\cite{pl04}, except for 
finding a point on each cycle of the merge curve in the merging step.

Let the input family  $S$ be subdivided into two subfamilies $\redClusters$ and 
$\blueClusters$ by a vertical line $\ell$ (using the location of the leftmost point in each cluster), and
let the diagrams $\redHVD$ and $\blueHVD$ be recursively computed. 
The main difference from the well-known divide-and-conquer algorithm for the Voronoi diagram of points~\cite{Shamos1975} 
 is that the merge curve
$\mergeCurve$ for $\redHVD$ and $\blueHVD$  may have several unbounded components and several cycles~\cite{pl04}. 
Merging    $\redHVD$ and $\blueHVD$ then  consists of (1) finding a point on each component of the merge curve, (2) tracing the merge curve and 
 stitching the relevant parts of  $\redHVD$ and $\blueHVD$ together.
It is known how to perform task (2) such that the overall time required for this task  
during the course of the algorithm is $O(n\log n + m)$~\cite{pl04}.
Task (1) for the unbounded components 
can be done in time $O(n)$, where $n$ is the total number 
of points in the clusters in $\redClusters$ and 
$\blueClusters$~\cite{pl04}. Each cycle of $\mergeCurve$ that encloses at least one vertex of $\redHVD$ or of $\blueHVD$, 
can be found in $O(n\log{n})$ time~\cite{pl04}.

We describe how to find a point on each cycle that encloses no vertices of $\redHVD$ or $\blueHVD$. 
	We make use of the fact that such cycles enclose a portion of an internal edge of one of the diagrams (which follows from 
	Property~\ref{prop:con-comp}). We preprocess $\redHVD$ and $\blueHVD$  as shown in Section~\ref{sec:ds}, resulting in a data structure, 
	that for every edge of $\redHVD$ provides a (implicit) tree storing all the 
	intersections between that edge and the edges of $\blueHVD$. 

	For each internal edge of $\redHVD$, we examine its portions that are outside any merge curve identified so far. 
	Let $uv$ be such a portion of an internal edge induced by $p1,p2 \in P, P \in \redClusters$. That is, 
	$uv$ 
	is a portion of the bisector of $p_1$ and $p_2$, and 
	both $u,v$ are closer to some cluster in $\blueClusters$ than to $P$.
	We must find all the portions of $uv$ that are closer to $P$ than
	to any cluster in $\blueClusters$, or report that they do not exist. 

	We search the tree that  stores the intersections between $uv$ and the edges of $\blueHVD$. 
As we navigate in the tree we keep track of the portion of $uv$ that corresponds to the current node:
the root of the tree corresponds to $uv$, and each consecutive  node $y$ on any root-to-leaf path subdivides 
the current portion into two, which are assigned to the two children of  $y$. 
Observe that no two distinct nodes of the tree
may be such that the portions corresponding to their left children intersect, and the ones of their right children intersect as well.
Indeed, for two nodes $x$ and $y$, either $x$ is an ancestor of $y$, and then both intervals of $y$ are contained in one of the intervals of $x$, 
or $y$ is an ancestor of $x$ (a symmetric case), or $x$ and $y$ do not lie on one root-to-leaf path, and then  the intervals of $x$
do not intersect the ones of $y$. 
 This observation will be useful for 
estimating the time complexity of the algorithm. 

Consider a tree node $x$, let $e$ be the edge of $\blueHVD$, such that $x$ is the intersection point between $uv$ and $e$,
and let $u'v'$ be the portion of $uv$ that corresponds to $x$. 
We need to decide in which subtree of $x$ to continue the search.  Suppose first that $e$ is an internal edge of $\blueHVD$. 
%
%
Let $Q$ be the cluster in $\blueClusters$, such that $e$ belongs to $\hregs{\blueClusters}{Q}$.  
	 Employing Lemma~\ref{lemma:oracle}, we determine if $u'x$ and/or $xv'$ contain points closer to $P$ than to $Q$. 
	 If $u'x$ (resp., $xv'$) does not contain such points, than 
it cannot intersect the merge curve, 
and the corresponding subtree should be ignored. Otherwise, the search continues in that subtree. 

 Suppose now that $e$ is a pure edge of $\blueHVD$, that separates the regions of  some clusters $Q,Q' \in \blueClusters$. 
	 We perform the checks for both $Q$ and $Q'$, and we continue to a subtree only if the corresponding portion of $u'v'$ 
contains points closer to $P$ that to $Q$ and points closer to $P$ than to $Q'$. 

We now analyze the time complexity of the merging procedure.
	Let $m_{\ell r}$ be the total number of crossings between pairs of clusters, one of which is in $\redClusters$ and the other is in $\blueClusters$. 
	All the unbounded components of the merge curve, as well as the cycles that contain Voronoi vertices, can be found in total time 
	$((n + m_{\ell r})\log n)$~\cite{pl04}. 
	When the merge curve is fully determined, 
stitching the appropriate pieces of $\redHVD$ and $\blueHVD$ together can be done within the same time bound~\cite{pl04}. 
Below we prove that our procedure  to find
	all the cycles of the merge curve that do not contain Voronoi vertices, requires $O((n+m_{\ell r})\log^2 n)$.  
The claimed overall time complexity to construct $\HVD{S}$ then will follow from the fact that for each pair of clusters in $P,Q \in S$, 
	it happens at most once that $P$ is in $\redHVD$ and $Q$ is in $\blueHVD$ or vice versa. Thus the total sum of the numbers 
	$m_{\ell r}$ in all the merging steps is at most $m$. 

	To prove that 
finding all empty cycles of the merge curve requires  $O((n+m_{\ell r})\log^2 n)$ time, 
we note that if at a node $x$ the search has continued in both subtrees of $x$, 
we necessarily have that the cluster $Q \in \blueClusters$ that is the closest to $x$, 
	such that both $u'x$ and $xv'$ contain points closer to $P$ than to $Q$. By Lemma~\ref{lemma:corr}, 
there are two crossing mixed vertices induced by $p1, p2$ and $Q$. We charge the branching of our search at $x$ to that pair of vertices. 
Since no two nodes of the tree may have the portions of $uv$ that correspond to their 
left children and the ones corresponding to their right children, respectively,  both intersecting (as observed above),  
 a pair of crossing vertices  can be charged at most once.
Therefore,  the total number of  nodes of all the search trees where our search went both ways
 is $O(m_{\ell r})$. The claim follows, since 
the number of leaves reached by our search
 in one tree is $1+t$, where $t$ is the number of times the search was continued to  both subtrees of a node. 
Since the height of any leaf is $O(\log{n})$, and the time spent in one node is $O(\log n)$, 
 the total time for searching one tree is $O((1 + t)\log^2{n})$. 
\end{proof}


\bibliography{rb}

\begin{thebibliography}{10}
\providecommand{\url}[1]{\texttt{#1}}
\providecommand{\urlprefix}{URL }

\bibitem{ahiklmps01}
Abellanas, M., Hurtado, F., Icking, C., Klein, R., Langetepe, E., Ma, L.,
  Palop, B., Sacrist{\'a}n, V.: The farthest color {V}oronoi diagram and
  related problems. In: 17th Eur. Workshop on Comput. Geom. pp. 113--116
  (2001), {T}ech. Rep. 002 2006, Univ. Bonn

\bibitem{aronov2006}
Aronov, B., Bose, P., Demaine, E.D., Gudmundsson, J., Iacono, J., Langerman,
  S., Smid, M.: Data structures for halfplane proximity queries and incremental
  voronoi diagrams. In: 7th Latin American Symp. on Theoretical Informatics
  (LATIN). pp. 80--92. Springer (2006)

\bibitem{chan94}
Chan, T.M.: A simple trapezoid sweep algorithm for reporting red/blue segment
  intersections. In: CCCG. pp. 263--268 (1994)

\bibitem{CEGS94}
Chazelle, B., Edelsbrunner, H., Guibas, L.J., Sharir, M.: Algorithms for
  bichromatic line-segment problems and polyhedral terrains. Algorithmica
  11(2),  116--132 (1994)

\bibitem{CKLP16}
Cheilaris, P., Khramtcova, E., Langerman, S., Papadopoulou, E.: A randomized
  incremental algorithm for the {H}ausdorff {V}oronoi diagram of non-crossing
  clusters. Algorithmica  76(4),  935--960 (2016)

\bibitem{CEGGHLLN11}
Cheong, O., Everett, H., Glisse, M., Gudmundsson, J., Hornus, S., Lazard, S.,
  Lee, M., Na, H.S.: Farthest-polygon voronoi diagrams. Comput. Geom.  44(4),
  234--247 (2011)

\bibitem{ckpss16-eurocg}
Claverol, M., Khramtcova, E., Papadopoulou, E., Saumell, M., Seara, C.:
  Stabbing circles for some sets of {D}elaunay segments. In: 32th Eur. Workshop
  on Comput. Geom. (EuroCG). pp. 139--143 (2016)

\bibitem{CKPSS17}
Claverol, M., Khramtcova, E., Papadopoulou, E., Saumell, M., Seara, C.:
  Stabbing circles for sets of segments in the plane. Algorithmica  (2017),
  {DOI} 10.1007/s00453-017-0299-z

\bibitem{dmt06}
Dehne, F., Maheshwari, A., Taylor, R.: A coarse grained parallel algorithm for
  {H}ausdorff {V}oronoi diagrams. In: Int. Conf. on Parallel Processing (ICPP).
  pp. 497--504. IEEE (2006)

\bibitem{dss86}
Driscoll, J.R., Sarnak, N., Sleator, D.D., Tarjan, R.E.: Making data structures
  persistent. In: 18th annual ACM Symp. on Theory of Computing. pp. 109--121.
  ACM (1986)

\bibitem{EGS89}
Edelsbrunner, H., Guibas, L.J., Sharir, M.: The upper envelope of piecewise
  linear functions: algorithms and applications. Discr. \& Comput. Geom.  4(4),
   311--336 (1989)

\bibitem{Guibas1989}
Guibas, L.J., Sharir, M., Sifrony, S.: {On the general motion-planning problem
  with two degrees of freedom}. Discr. {\&} Comput. Geom.  4(5),  491--521
  (1989)

\bibitem{hks1993dcg}
Huttenlocher, D.P., Kedem, K., Sharir, M.: The upper envelope of {V}oronoi
  surfaces and its applications. Discr. \& Comput. Geom.  9(3),  267--291
  (1993)

\bibitem{kp16a}
Khramtcova, E., Papadopoulou, E.: Randomized incremental construction for the
  {H}ausdorff {V}oronoi diagram of point clusters. ArXiv e-prints  (2016),
  {arXiv:1612.01335}

\bibitem{ms1988}
Mairson, H.G., Stolfi, J.: Reporting and counting intersections between two
  sets of line segments. In: Theoretical Foundations of Computer Graphics and
  CAD, pp. 307--325 (1988)

\bibitem{ms01}
Mantler, A., Snoeyink, J.: Intersecting red and blue line segments in optimal
  time and precision. In: Discr. and Comput. Geom.: Japanese Conference, JCDCG
  2000, Revised Papers. pp. 244--251 (2001)

\bibitem{M01}
Mehlhorn, K., Meiser, S., Rasch, R.: {Furthest site abstract Voronoi diagrams}.
  Int. J. Comput. Geom. {\&} Appl.  11(06),  583--616 (2001)

\bibitem{M87}
Mount, D.M.: Storing the subdivision of a polyhedral surface. Discr. \& Comput.
  Geom.  2(2),  153--174 (1987)

\bibitem{ps94}
Palazzi, L., Snoeyink, J.: Counting and reporting red/blue segment
  intersections. CVGIP: Graphical Models and Image Processing  56(4),  304--310
  (1994)

\bibitem{p04}
Papadopoulou, E.: The {H}ausdorff {V}oronoi diagram of point clusters in the
  plane. Algorithmica  40(2),  63--82 (2004)

\bibitem{pl04}
Papadopoulou, E., Lee, D.T.: The {H}ausdorff {V}oronoi diagram of polygonal
  objects: A divide and conquer approach. Int. J. of Comput. Geom. \& Appl.
  14(06),  421--452 (2004)

\bibitem{Shamos1975}
Shamos, M.I., Hoey, D.: {Closest-point problems}. In: 16th Annual Symp. on
  Foundations of Computer Science ({SFCS}). pp. 151--162. IEEE (1975)

\bibitem{S88}
Sharir, M.: The shortest watchtower and related problems for polyhedral
  terrains. Inf. Process. Lett.  29(5),  265--270 (1988)

\bibitem{Z97}
Zhu, B.: Computing the shortest watchtower of a polyhedral terrain in
  {$O(n\log{n})$} time. Comput. Geom.  8(4),  181--193 (1997)

\end{thebibliography}
\end{document}